\documentclass{tlp}
\pdfoutput=1
\usepackage{amssymb}
\usepackage{amsmath}
\usepackage{url}
\usepackage{color}

\def\eqdef{\stackrel{\rm def}{=} \;}

\def\is{:=}
\def\u{\hbox{\tt u}}
\def\P{{\cal P}}
\def\C{{\cal C}}
\def\F{{\cal F}}
\def\T{{\cal T}}
\def\h{h}
\def\t{t}
\def\SQHT{{\hbox{SQHT}^=}}
\def\SQHTF{{\hbox{SQHT}^=_\F}}
\def\QELF{{\hbox{QEL}^=_\F}}
\def\QEL{\hbox{QEL}}

\newcommand{\tuple}[1]{\langle #1 \rangle}

\newcommand{\V}[1]{\mathbf{#1}}
\newcommand{\tr}[1]{#1^*}

\def\IF{\ \hbox{\tt :-} \ }

\def\ap{\; \# \; }

\def\qed{~\hfill$\Box$}

\newtheorem{definition}{Definition}
\newtheorem{theorem}{Theorem}
\newtheorem{example}{Example}
\newtheorem{lemma}{Lemma}

\newtheorem{proposition}{Proposition}

\newcommand{\replace}[2]{[#1 / #2]}

\begin{document}
\bibliographystyle{acmtrans}

\submitted{21 June 2009}
\revised{}
\accepted{}
\title{Functional Answer Set Programming}

\author[P. Cabalar]
{
Pedro Cabalar \\
Department of Computer Science, \\
University of Corunna, Spain. \\
email: \tt cabalar@udc.es
}


\maketitle

\begin{abstract}

\end{abstract}
\begin{keywords}
Answer Set Programming, Equilibrium Logic, Partial Functions, Functional Logic Programming.
\end{keywords}

\section{Introduction}

\label{sec:intro}

Since its introduction two decades ago~\cite{GL88}, the paradigm of \emph{Answer Set Programming} (ASP)~\cite{MT99} has gradually become one of the most successful and practical formalisms for Knowledge Representation due to its flexibility, expressiveness and current availability of efficient solvers. This success can be easily checked by the continuous and plentiful presence of papers on ASP in the main conferences and journals on Logic Programming, Knowledge Representation and Artificial Intelligence during the last years. The declarative semantics of ASP has allowed many syntactic extensions that have simplified the formalisation of complex domains in different application areas like constraint satisfaction problems, planning or diagnosis.

In this paper we consider one more syntactic extension that is an underlying feature in most application domains: the use of \emph{(partial) evaluable functions}. Most ASP programs include some predicates that are nothing else than relational representations of functions from the original domain being modelled. For instance, when modelling the typical educational example of family relationships, we may use a predicate $mother(X,Y)$ to express that $X$'s mother is $Y$, but of course, we must add an additional constraint to ensure that $Y$ is unique wrt $X$, i.e., that the predicate actually acts as the function $mother(X)=Y$. In fact, it is quite common that first time Prolog students use this last notation as their first attempt. Functions are not only a natural element for knowledge representation, but can also simplify in a considerable way ASP programs. Apart from avoiding constraints for uniqueness of value, the possibility of nesting functional terms like in $W=mother(father(mother(X)))$ allows a more compact and readable representation than the relational version $mother(X,Y), father(Y,Z), mother(Z,W)$ involving extra variables, which may easily mean a source of formalisation errors. Similarly, as we will see later, the use of partial functions can also save the programmer from including explicit conditions in the rule bodies to check that the rule head is actually defined.

The addition of functions to ASP is not new at all. In fact, there exist two different ways in which functions are actually understood. The first way of treating functions is followed by most of the existing work in the topic (like the general approaches~\cite{Syr01,Bon04,SE07} or the older use of function $Result$ for Situation Calculus inside ASP~\cite{GL93}). These approaches treat functions in the same way as Prolog, that is, they are just a way for \emph{constructing} the Herbrand universe, and so they satisfy the unique names assumption -- e.g. $mother(john)=mary$ is always false. A second way of treating functions is dealing with them as in Predicate Calculus, as done for instance in \emph{Functional Logic Programming}~\cite{Han94}. The first and most general approach in this direction is due to the logical characterisation of ASP in terms of \emph{Equilibrium Logic}~\cite{Pea96} and, in particular, to its extension to first order theories, \emph{Quantified Equilibrium Logic} (QEL)~\cite{PV04}. As a result of this characterisation, the concept of stable model is now defined for any theory from predicate calculus with equality. In fact, stable models can be alternatively described by a second-order logic operator~\cite{FLL07} quite close to Circumscription \cite{McC80}, something that has been already used, for instance, to study strong equivalence for programs with variables~\cite{LPV07}. Another alternative for ASP with (non-Herbrand) functions has been very recently presented in~\cite{LW08} and, as we will show later, can be seen as a particular case of the current approach, when we restrict to total functions.

As we will explain in the next section, we claim that the exclusive use of Herbrand functions and the currently proposed interpretation of equality in QEL or the requirement for functions to be total, as in~\cite{LW08}, yield some counterintuitive results when introducing functions for knowledge representation. In order to overcome these problems, we propose a variation of QEL that separates Herbrand functions (or constructors) from evaluable functions, as also done in logical characterisations~\cite{GHLR99,Rod01,Han07} of \emph{Functional Logic Programming}. We further show how our semantics for partial functions has a direct relation to the \emph{Logic of Existence} (or $E$-logic) proposed by Scott~\cite{Sco79}.

The rest of the paper\footnote{This paper extends~\cite{Cab08} and improves it in many different ways. The most significant are, firstly, the inclusion of Section~\ref{sec:lin} with a complete formal comparison to~\cite{LW08} plus a small discussion on expressiveness. Second, the safety condition has been corrected (some cases dealing with equality were wrong) and improved to cover more cases. Third, all proofs have been completed now and included in an appendix.} is organized as follows. In the next section, we informally consider some examples of knowledge representation with functions in ASP, commenting the apparently expected behaviour and the problems that arise when using the current proposal for QEL. In Section~\ref{sec:eq}, we introduce our variant called $\QELF$. Section~\ref{sec:der} defines some useful derived operators, many of them directly extracted from $E$-logic and showing the same behaviour. In Section~\ref{sec:FLP} we consider a syntactic subclass of logic programs with evaluable functions and Herbrand constants, and show how they can be translated into (non-functional) normal logic programs afterwards. This includes a definition of safety that guarantees that the final translation results in a safe program, something crucial for the current ASP grounders. Section~\ref{sec:lin} establishes a formal comparison showing how~\cite{LW08} can be encoded into our functional logic programs by forcing functions to be total, and also includes a discussion showing that $\QELF$ is more suitable for nonmontonic reasoning with functions. Finally, Section~\ref{sec:rel} contains a brief discussion about other related work and Section~\ref{sec:conc} concludes the paper.

\section{A Motivating Example}
\label{sec:mot}

Consider the following simple scenario with a pair of rules.
\begin{example}
\label{ex:1}
When deciding the second course of a given meal once the first course is fixed, we want to apply the following criterion: on Fridays, we repeat the first course as second one; the rest of week days, we choose $fish$ if the first was $pasta$.\qed
\end{example}
A straightforward encoding of these rules\footnote{As a difference wrt to the typical ASP notation, we use $\neg$ to represent default negation and, instead of a comma, we use $\wedge$ to separate literals in the body.} into ASP would correspond to the program $\Pi_1$:
\begin{eqnarray}
second(fish) & \leftarrow & first(pasta) \wedge \neg friday \label{f:meal1}\\
second(X) & \leftarrow & first(X) \wedge friday \label{f:meal2}\\
\bot & \leftarrow & first(X) \wedge first(Y) \wedge X\neq Y \\
\bot & \leftarrow & second(X) \wedge second(Y) \wedge X\neq Y \hspace{15pt}
\end{eqnarray}
\noindent where the last two rules just represent that each course is unique, i.e., $first(salad)$ and $first(pasta)$ cannot be simultaneously true, for instance. In fact, these constraints immediately point out that $first$ and $second$ are 0-ary functions. A very naive attempt to use these functions for representing our example problem could be the pair of formulas $\Pi_2$:
\begin{eqnarray}
second=fish & \leftarrow & first=pasta \wedge \neg friday \label{f:mealb1}\\
second=first & \leftarrow & friday \label{f:mealb2}
\end{eqnarray}

\noindent Of course, $\Pi_2$ is not a logic program, but it can still be given a logic programming meaning by interpreting it under Herbrand models of QEL, or the equivalent recent characterisation of stable models for first order theories~\cite{FLL07}. Unfortunately, the behaviour of $\Pi_2$ in QEL with Herbrand models (and decidable equality) will be quite different to that of $\Pi_1$ by several reasons that can be easily foreseen. First of all, there exists now a qualitative difference between functions $first$ and $second$ with respect to $fish$ and $pasta$. For instance, while it is clear that $fish=pasta$ must be false, we should allow $second=first$ to cope with our Fridays criterion. If we deal with Herbrand models or unique names assumption, the four constants would be pairwise different and  (\ref{f:mealb1}) would be equivalent to $\bot \leftarrow \bot$, that is, a tautology, whereas (\ref{f:mealb2}) would become the constraint $\bot \leftarrow friday$. 

Even after limiting the unique names assumption only to constants $fish$ and $pasta$, new problems arise. For instance, the approaches in~\cite{PV04,FLL07,LPV07,LW08} deal with complete functions and the axiom of \emph{decidable equality}:
\begin{align}
x=y \vee \neg (x=y) \tag{DE}
\end{align}
\noindent This axiom is equivalent to $x=y \leftarrow \neg \neg (x=y)$ which informally implies that we always have a justification to assign any value to any function. 
Thus, for instance, if it is not Friday and we do not provide any information about the first course, i.e., no atom $first(X)$ holds, then $\Pi_1$ will not derive any information about the second course, that is, no atom $second(X)$ is derived. In $\Pi_2$, however, functions $first$ and $second$ must \emph{always} have a value, which is further justified in any stable model by (DE). As a result, we get that a possible stable model is, for instance, $first=fish$ and $second=pasta$. A related problem of axiom (DE) is that it allows rewriting a rule like (\ref{f:mealb1}) as the constraint:
\begin{eqnarray*}
\bot & \leftarrow & first=pasta \wedge \neg friday \wedge \neg (second=fish)
\end{eqnarray*}
\noindent whose relational counterpart would be
\begin{eqnarray}
\bot & \leftarrow & first(pasta) \wedge \neg friday \wedge  \neg second(fish) \ \ \label{f:meal1bis}
\end{eqnarray}
\noindent and whose behaviour in logic programming is very different from the original rule (\ref{f:meal1}). As an example, while $\Pi_1 \cup \{first(pasta)\}$ entails $second(fish)$, the same program after replacing (\ref{f:meal1}) by (\ref{f:meal1bis}) has no stable models.

Finally, even after removing decidable equality, we face a new problem that has to do with directionality in the equality symbol when used in the rule heads. The symmetry of `=' allows rewriting (\ref{f:mealb2}) as:
\begin{eqnarray}
first=second & \leftarrow & friday \label{f:mealb2bis}
\end{eqnarray}
\noindent that in a relational notation would be the rule:
\begin{eqnarray}
first(X) & \leftarrow & second(X) \wedge friday \label{f:meal2bis}
\end{eqnarray}
\noindent which, again, has a very different meaning from the original (\ref{f:meal2}). For instance $\Pi_1 \cup \{friday, second(fish)\}$ does not entail anything about the first course, whereas if we replace in this program (\ref{f:meal2}) by (\ref{f:meal2bis}), we obtain $first(fish)$. This is counterintuitive, since our program was intended to derive facts about the second course, and not about the first one. To sum up, we will need some kind of new directional operator to specify the function value in a rule head.

\section{Quantified Equilibrium Logic with Evaluable Functions}
\label{sec:eq}

The definition of propositional Equilibrium Logic~\cite{Pea96} relied on establishing a selection criterion on models of the intermediate logic, called the logic of \emph{Here-and-There} (HT)~\cite{Hey30}. The first order case~\cite{PV04} followed similar steps, introducing a quantified version of HT, called $\SQHT$ that stands for \emph{Quantified HT with static domains\footnote{The term \emph{static domain} refers to the fact that the universe is shared among all worlds in the Kripke frame.} and equality}.  In this section we describe the syntax and semantics of a variant, called $\SQHTF$, for dealing with evaluable functions. 

We begin by defining a first-order language by its \emph{signature}, a tuple $\Sigma=\tuple{\C,\F,\P}$ of disjoint sets where $\C$ and $\F$ are sets of \emph{function names} and $\P$ a set of \emph{predicate names}. We assume that each function (resp. predicate) name has the form $f/n$ where $f$ is the function (resp. predicate) symbol, and $n\geq 0$ is an integer denoting the number of arguments (or \emph{arity}). Elements in $\C$ will be called \emph{Herbrand functions} (or \emph{constructors}), whereas elements in $\F$ will receive the name of \emph{evaluable\footnote{In~\cite{Han07}, elements of $\F$ are called \emph{defined functions} instead -- we avoid this terminology because it could be mistakenly understood as the opposite of being \emph{undefined} or partial.} functions} (or \emph{operations}). The sets $\C_0$ (Herbrand constants) and $\F_0$ (evaluable constants) respectively represent the elements of $\C$ and $\F$ with arity $0$. We assume $\C_0$ contains at least one element.

First-order formulas are built up in the usual way, with the same syntax of classical predicate calculus with equality $=$. We assume that $\neg \varphi$ is defined as $\varphi \rightarrow \bot$ whereas $x \neq y$ just stands\footnote{We hope that, depending on the context, the reader will be aware of the different use of symbols `=' and `$\neq$' as formulas in the language from their standard use in the semantic metalanguage.} for $\neg (x=y)$. An atom like $t=t'$ is called an \emph{equality atom}, whereas an atom like $p(t_1,\dots,t_n)$ with $n\geq 0$ for any predicate $p$ different from equality receives the name of \emph{predicate atom}. Given any set of functions ${\cal A}$ we write $Terms({\cal A})$ to stand for the set of ground terms built from functions (and constants) in ${\cal A}$. In particular, the set of all possible ground terms for signature $\Sigma=\tuple{\C,\F,\P}$ would be $Terms(\C \cup \F)$ whereas the subset $Terms(\C)$ will be called the \emph{Herbrand Universe} of $\cal L$. The \emph{Herbrand Base} $HB(\C,\P)$ is a set containing all atoms that can be formed with predicates in $\P$ and terms in the Herbrand Universe, $Terms(\C)$.

From now on, we assume that all free variables are implicitly universally quantified. We use letters $x, y, z$ and their capital versions to denote variables, $t$ to denote terms, and letters $c, d$ to denote ground terms. Boldface letters like $\V{x}, \V{t}, \V{c}, \dots$ represent tuples (in this case of variables, terms and ground terms, respectively). The corresponding semantics for $\SQHTF$ is described as follows. 

\begin{definition}[state]
\label{def:state}
A \emph{state} for a signature $\Sigma=\tuple{\C,\F,\P}$ is a pair $(\sigma,A)$ where $A \subseteq HB(\C,\P)$ is a set of atoms from the Herbrand Base and \\
$\sigma : Terms(\C \cup \F) \rightarrow Terms(\C) \cup \{ \u \}$ is a function assigning to any ground term in the language some ground term in the Herbrand Universe or the special value $\u \not\in Terms(\C \cup \F)$ (standing for \emph{undefined}). Function $\sigma$ must satisfy:
\begin{enumerate}
\item[\rm (i)] $\sigma(c) = c$ for all $c \in Terms(\C)$.
\item[\rm (ii)] $\sigma(f(t_1,\dots,t_n)) = \left\{
\begin{array}{ll}
\u & \hbox{if } \sigma(t_i)=\u \ \hbox{for some } i=1\dots n \\
\sigma(f(\sigma(t_1),\dots,\sigma(t_n))) & \hbox{otherwise} 
\end{array}
\right.$
\end{enumerate}
\qed
\end{definition}

As we can see, our domain is exclusively formed by the terms from the Herbrand Universe, $Terms(\C)$. These elements are used as arguments of ground atoms in the set $A$, that collects the \emph{true} atoms in the state. Similarly, the value of any functional term is an element from $Terms(\C)$, excepting the cases in which operations are left undefined (i.e., they are \emph{partial} functions) -- if so, they are assigned the special element $\u$ (different from any syntactic symbol) instead. Condition (i) asserts, as expected, that any term $c$ from the Herbrand Universe has the fixed valuation $\sigma(c)=c$. Condition (ii) establishes two important restrictions. On the one hand, it guarantees that a functional term with an undefined argument becomes undefined in its turn\footnote{Using Functional Logic Programming terminology, this means that functions are \emph{strict}, as opposed to non-strict functions with lazy evaluation.}. On the other hand, Condition (ii) also guarantees that functions preserve their interpretation through subterms -- for instance, if we have $\sigma(f(a))=c$ we expect that $\sigma(g(f(a))$ and $\sigma(g(c))$ coincide. It is easy to see that (ii) implies that $\sigma$ is completely determined by the values it assigns to all terms like $f(\V{c})$ where $f$ is any operation and $\V{c}$ a tuple of elements in $Terms(\C)$.

\begin{definition}[Ordering $\preceq$ among states]
\label{def:order}
We say that state $S=(\sigma,A)$ is \emph{smaller} than state $S'=(\sigma',A')$, written $S \preceq S'$, when both:
\begin{enumerate} 
\item[i)] $A \subseteq A'$. 
\item[ii)] $\sigma(d)=\sigma'(d)$ or $\sigma(d)=\u$, for all $d \in Terms(\C \cup \F)$.
\qed
\end{enumerate}
\end{definition}
We write $S \prec S'$ when the relation is strict, that is, $S\preceq S'$ and $S \neq S'$. The intuitive meaning of $S \preceq S'$ is that the former contains \emph{less information} than the latter, so that any true atom or defined function value in $S$ must hold in $S'$.

\begin{definition}[$HT$-interpretation]
An $HT$ \emph{interpretation} $I$ for a signature $\Sigma=\tuple{\C,\F,\P}$ is a pair of states $I=\tuple{S^\h,S^\t}$ with $S^\h \preceq S^\t$.\qed
\end{definition}

The superindices $\h,\t$ represent two worlds (respectively standing for \emph{here} and \emph{there}) with a reflexive ordering relation further satisfying $\h \leq \t$. An interpretation like $\tuple{S^\t,S^\t}$ is said to be \emph{total}, referring to the fact that both states contain the same information\footnote{Note that by \emph{total} we do not mean that functions cannot be left undefined. We may still have some term $d$ for which $\sigma^\t(d)=\u$.}.

Given an interpretation
$I=\tuple{S^\h, S^\t}$, with $S^\h=(\sigma^\h,I^\h)$ and $S^\t=(\sigma^\t,I^\t)$, we define when $I$ \emph{satisfies} a formula $\varphi$ at some world $w\in \{h,t\}$, written $I,w \models \varphi$, inductively as follows:
\begin{itemize}
\item $I,w \models p(t_1,\dots,t_n)$ if $p( \sigma^w(t_1), \dots, \sigma^w(t_n) ) \in I^w$;
\item $I,w \models t_1=t_2$ if $\sigma^w(t_1)=\sigma^w(t_2) \neq \u$;
\item $I,w \not\models \bot$; $I,w \models \top$;
\item $I,w \models \alpha \wedge \beta$ if $I,w \models \alpha$ and $I,w \models \beta$;
\item $I,w \models \alpha \vee \beta$ if $I,w \models \alpha$ or $I,w \models \beta$;
\item $I,w \models \alpha \rightarrow \beta$ if for all $w'\geq w$: $I,w' \not\models \alpha$ or $I,w' \models \beta$;
\item $I,w \models \forall x \ \alpha(x)$ if for each $w'\geq w$ and each $c \in Terms(\C)$: $I,w' \models \alpha(c)$;
\item $I,w \models \exists x \ \alpha(x)$ if for some $c \in Terms(\C)$: $I,w \models \alpha(c)$.
\qed
\end{itemize}

An important observation is that the first condition above implies that an atom with an undefined argument will always be valuated as false since, by definition, $\u$ never occurs in ground atoms of $I^\h$ or $I^\t$. Something similar happens with equality: $t_1 = t_2$ will be false if any of the two operands, or even both, are undefined. As usual, we say that $I$ is a \emph{model} of a formula $\varphi$, written $I \models \varphi$, when $I,h \models \varphi$.  Similarly, $I$ is a \emph{model} of a theory $\Gamma$ when it is a model of all of its formulas.

From the definition of $\neg$ as derived operator, we can easily check that:
\begin{proposition}\label{prop:neg}
$I,w \models \neg \varphi$ iff $I,t \not\models \varphi$.\qed
\end{proposition}

Nonmonotonicity is obtained by the next definition, which introduces the idea of equilibrium models for $\SQHTF$.

\begin{definition}[Equilibrium model]
A model $\tuple{S^t,S^t}$ of a theory $\Gamma$ is an \emph{equilibrium model} if there is no strictly smaller state $S^h \prec S^t$ that $\tuple{S^h,S^t}$ is also model of $\Gamma$. \qed
\end{definition}

The Quantified Equilibrium Logic with evaluable functions ($\QELF$) is the logic induced by the $\SQHTF$ equilibrium models. 

For space reasons we describe $\SQHT$ (resp. QEL) as a particular instance of $\SQHTF$ (resp. $\QELF$). It can be easily checked that this description is equivalent to the one in~\cite{PV08}. The syntax for $\SQHT$ is the same as for $\SQHTF$ (that is, Predicate Calculus with equality) but starting from a signature $\tuple{\F,\P}$ where no distinction is made among functions in set $\F$. Each $\SQHT$ interpretation for signature $\tuple{\F,\P}$ further deals with a universe domain, a set $U\neq \emptyset$ which is said to be \emph{static}, that is, common to both worlds $\h$ and $\t$. To capture this in $\SQHTF$ we can just use signature $\tuple{\C,\F,\P}$ and define $\C$ as a set of constant names, one $c'$ per each individual $c \in U$. The most important feature of $\SQHT$ interpretations is that they satisfy the axiom $t=t$ for any term $t$. In other words, any ground term $d \in Terms(\C \cup \F)$ is defined $\sigma^\h(d) \neq \u$ and, in fact, by construction of interpretations, this also means $\sigma^\h(d)=\sigma^\t(d)$. As a result, $\hbox{SQHT}^=$ actually uses a unique $\sigma$ function for both worlds $\h$ and $\t$ and interpretations can be represented instead as $\tuple{\sigma,I^\h,I^\t}$. Under this restriction, it is easy to see that decidable equality $t_1=t_2 \vee t_1\neq t_2$ is a valid formula.

Herbrand models from $\hbox{SQHT}^=$ and signature $\tuple{\C,\P}$ can be easily captured by just considering $\SQHTF$ interpretations for signature $\tuple{\C,\emptyset,\P}$. Finally, the models selection criterion in the definition of equilibrium models need not be modified. Since $\sigma^\h=\sigma^t=\sigma$ and all terms and defined, the $\preceq$ ordering relation among states in QEL actually amounts to a simple inclusion of sets of ground atoms.


\section{Useful Derived Operators}
\label{sec:der}

From the $\SQHTF$ semantics, it is easy to see that the formula $(t=t)$, usually included as an axiom for equality, is not valid in $\SQHTF$. In fact, $I,w \models (t=t)$ iff $\sigma^w(t)\neq \u$, that is, term $t$ is defined. In this way, we can introduce Scott's~\cite{Sco79} \emph{existence} operator\footnote{Contrarily to the original Scott's $E$-logic, variables in $\SQHTF$ are always defined. This is not an essential difference: terms may be left undefined instead, and so most theorems, like $(x=y)\rightarrow(y=x)$ are expressed here using metavariables for terms $(t_1=t_2) \rightarrow (t_2=t_1)$.} in a standard way: $
E \ t \eqdef (t=t)$.
Condition (ii) in Definition~\ref{def:state} implies the \emph{strictness} condition of $E$-logic, formulated by the axiom $E \ f(t) \rightarrow E \ t$. As happens with $(t=t)$, the substitution axiom for functions:
\begin{eqnarray*}
t_1=t_2 \rightarrow f(t_1)=f(t_2)
\end{eqnarray*}
\noindent is not valid, since it may be the case that the function is undefined. However, the following weaker version is an $\SQHTF$ tautology:
\begin{eqnarray*}
t_1=t_2 \wedge E \ f(t_1) & \rightarrow & f(t_1)=f(t_2)
\end{eqnarray*}
Usual axioms for equality that are valid in $\SQHTF$ are, for any predicate $P$:
\begin{eqnarray*}
t_1=t_2 & \rightarrow & t_2=t_1\\
t_1=t_2 \wedge t_2=t_3 & \rightarrow & t_1=t_3\\
t_1=t_2 \wedge P(t_1) & \rightarrow & P(t_2)
\end{eqnarray*}

At this point, it is perhaps convenient to introduce a few terms to talk about particular types of functions. We say that an evaluable function $f$ is \emph{decidable} under a given interpretation $I$, when $I$ satisfies the excluded middle axiom:
\begin{eqnarray}
f(\V{t})=t' \vee f(\V{t}) \neq t' \label{f:EMf}
\end{eqnarray}
\noindent and we say that $f$ is \emph{decidable} in a given theory when it is decidable under any of its models. The models of \eqref{f:EMf} correspond to interpretations where $\sigma^h(f(\V{c}))=\sigma^t(f(\V{c}))$ for any tuple of elements $\V{c} \in Terms(\C)$, that is, function $f$ has the \emph{same} interpretation in both worlds, and so, it somehow behaves ``classically.'' As example of decidable functions, think about integer arithmetic operations like +, -, $\times$ or $\div$, for which we expect a fixed interpretation in worlds $h$ and $t$ with their usual meaning. Of course, an evaluable function is always decidable under any equilibrium model $I$ of a theory $\Gamma$, since $\sigma^h=\sigma^t$ in that case, but this does not necessarily mean that \eqref{f:EMf} holds for all $\SQHTF$ models of $\Gamma$.

A function is said to be \emph{total} under an interpretation $I$ when $I$ satisfies:
\begin{eqnarray}
E \ \V{t} \rightarrow E \ f(\V{t}) \label{f:totalf}
\end{eqnarray}
\noindent and called \emph{partial} under $I$ otherwise. We say that a decidable function $f$ is \emph{total} in a theory $\Gamma$ if it is total under any of its models; otherwise it is \emph{partial} in $\Gamma$. Semantically, total functions satisfy $\sigma^w(f(\V{c})) \neq \u$ for any tuple of terms $\V{c} \in Terms(\C)$ and any world $w$, while partial functions no. From this, and the $\SQHTF$ semantics, it can be observed that \eqref{f:totalf} actually implies \eqref{f:EMf}, that is, a total function is always decidable. The opposite does not necessarily hold. Back to the example, under their usual interpretation, +, - and $\times$ are total functions, whereas $\div$ is partial although, as we said before, still decidable; in particular the formula $E \ x \rightarrow E \ x \div 0$, is always false, since a variable $x$ is always defined, $E \ x$, whereas $\sigma^w(x \div 0)=\u$. Again, a function may be total under equilibrium models of a theory $\Gamma$ without being so under any $\SQHTF$ model of $\Gamma$.

Similarly to these distinctions among types of functions, we have now several types of equalities and inequalities. In $E$-logic there exists a second and weaker equality (which Scott called \emph{equivalence}) that can be defined as $t_1 \equiv t_2 \eqdef (E \ t_1 \vee E \ t_2) \rightarrow t_1=t_2$. In other words, $t_1$ and $t_2$ have the same defined value, provided that any of them is defined. This equality is perhaps not so interesting for knowledge representation and its inclusion in logic programs, but may have a crucial importance when studying properties of programs, like for instance, strongly equivalent\footnote{Two theories $\Gamma_1, \Gamma_2$ are said to be \emph{strongly equivalent} when, for any theory $\Gamma$, the equilibrium models of $\Gamma_1 \cup \Gamma$ and $\Gamma_2 \cup \Gamma$ coincide.} transformations. In this sense, $t_1 \equiv t_2$ means that $t_1$ can be replaced by $t_2$ and vice versa, that is, they have the \emph{same behaviour}. For instance, the following valid formula:
\begin{eqnarray*}
t_1 \equiv t_2 \rightarrow f(t_1) \equiv f(t_2)
\end{eqnarray*}
\noindent allows us replacing $f(t_1)$ by $f(t_2)$ when we know that $t_1$ can be replaced by $t_2$. Note that, if we just use equality for that purpose
\begin{eqnarray*}
t_1 = t_2 \rightarrow f(t_1) = f(t_2)
\end{eqnarray*}
\noindent we would be forcing $f$ to become a total function, since taking $t_2$ to be $t_1$ above we actually get $E \ t_1 \rightarrow E \ f(t_1)$, and this is not what we want. To understand the difference, note that $t=0 \rightarrow (1 \div t)=(1 \div 0)$ is always false because $(1 \div 0)$ is undefined, whereas $t \equiv 0 \rightarrow (1 \div t) \equiv (1 \div 0)$ is valid.

To represent the \emph{difference} between two terms, we may also have several alternatives. The straightforward one is just $\neg (t_1=t_2)$, or abbreviated $t_1 \neq t_2$. However, this formula can be satisfied when any of the two operands is undefined. We may sometimes want to express a stronger notion of difference that behaves as a positive formula (this is usually called \emph{apartness} in the intuitionistic literature~\cite{Hey56}). In our case, we are especially interested in an apartness operator $t_1 \ap t_2$ where both arguments are required to be defined:
\begin{eqnarray*}
t_1 \ap t_2 & \eqdef & E \ t_1 \wedge E \ t_2 \wedge \neg (t_1=t_2)
\end{eqnarray*}
\noindent The semantic effect of this operator is that $I,w \models t_1 \ap t_2$ iff $\sigma^w(t_1) \neq \u$, $\sigma^w(t_2) \neq \u$ and $\sigma^w(t_1) \neq \sigma^w(t_2)$. To understand its meaning, consider the difference between $\neg (King(France)=LouisXIV)$ and $King(Spain)Ê\ap LouisXIV$. The first expression means that we cannot prove that the King of France is Louis XIV, what includes the case in which France has not a king. The second expression means that we can prove that the King of Spain (and so, such a concept exists) is not Louis XIV.

The next operator we introduce has to do with definedness of rule heads in logic programs (as far as we know, it has not been considered in the literature). The inclusion of a formula in the consequent of an implication may have an undesired effect when thinking about its use as a rule head. For instance, consider the rule $visited(next(x)) \leftarrow visited(x)$ and assume we have the fact $visited(1)$ but there is no additional information about $next(1)$. We would expect that the rule above does not yield any particular effect on $next(1)$. Unfortunately, as $visited(next(1))$ must be true, the function $next(1)$ must become defined and, as a collateral effect, it will be assigned some arbitrary value, say $next(1)=10$ so that $visited(10)$ is made true. To avoid this problem, we will use a new operator $\IF$ to define a different type of implication where the consequent is only forced to be true when all the functional terms that are ``necessary to build'' the atoms in the consequent are defined. Given a term $t$ we define its set of \emph{structural arguments} $Args(t)$ as follows:
\begin{itemize}
\item $Args(t) \eqdef \{t_1,\dots,t_n\}$ if $t$ has the form $f(t_1,\dots,t_n)$ for any evaluable function $f/n \in \F$. 
\item $Args(t) \eqdef t$ otherwise. 
\end{itemize}

\noindent We extend this definition for any atom $A$, so that its set of structural arguments $Args(A)$ corresponds to:
\begin{eqnarray*}
Args(P(t_1,\dots,t_n)) & \eqdef & \{t_1,\dots,t_n\} \\
Args(t = t') & \eqdef & Args(t) \cup Args(t')
\end{eqnarray*}
\noindent In our previous example, $Args(visited(next(x)))=\{next(x)\}$. Notice that, for an equality atom $t = t'$, we do not consider $\{t,t'\}$ as arguments as we have done for the rest of predicates, but go down one level instead, considering $Args(t) \cup Args(t')$ in its turn. For instance, if $A$ is the atom $friends(mother(x),mother(y))$, then $Args(A)$ would be $\{mother(x),mother(y)\}$, whereas for an equality atom $A'$ like $mother(x)=mother(y)$, $Args(A')=\{x,y\}$. We define $[\varphi]$ as the result of replacing each atom $A$ in $\varphi$ by the conjunction of all $E \ t \rightarrow A$ for each $t \in Args(A)$. We can now define the new implication operator as follows $\varphi \IF \psi \eqdef \psi \rightarrow [\varphi]
$. Back to the example, if we use now $visited(next(x)) \IF visited(x)$ we obtain, after applying the previous definitions, that it is equivalent to:
\begin{eqnarray*}
& & visited(x) \rightarrow [ visited(next(x))Ê]Ê\\
& \leftrightarrow & visited(x) \rightarrow (E \ next(x) \rightarrow visited(next(x)))Ê\\
& \leftrightarrow & visited(x) \wedge E \ next(x) \rightarrow visited(next(x))
\end{eqnarray*}

Another important operator will allow us to establish a direction in a rule head assignment -- remember the discussion about distinguishing between (\ref{f:mealb2}) and (\ref{f:mealb2bis}) in Section~\ref{sec:mot}. We define this \emph{assignment} operator as follows:
\begin{eqnarray*}
f(\V{t}) \is t' & \eqdef & E \ t' \rightarrow f(\V{t})=t'
\end{eqnarray*}
\noindent Now, our Example~\ref{ex:1} would be encoded with the pair of formulas:
\begin{eqnarray*}
second \is fish \IF first=pasta \wedge \neg friday \hspace{30pt} 
second \is first \IF friday 
\end{eqnarray*}
\noindent that, after some elementary transformations, lead to:\vspace{-5pt}
\begin{eqnarray*}
second=fish \leftarrow first=pasta \wedge \neg friday \\
second=first \leftarrow E \ first \wedge friday 
\end{eqnarray*}

Using these operators, a compact way to fix a default value $t'$ for a function $f(\V{t})$ would be $f(\V{t}) := t' \IF \neg (f(\V{t}) \ap t')$. Finally, we introduce a nondeterministic choice assignment with the following set-like expression:
\begin{eqnarray}
f(\V{t}) \in  \{ x \ | \ \varphi(x) \} \label{f:set}
\end{eqnarray}
\noindent where $\varphi(x)$ is a formula (called the set \emph{condition}) that contains the free variable $x$. The intuitive meaning of~(\ref{f:set}) is self-explanatory. As an example, the formula $
a \in \{ x \ | \ \exists y \ Parent(x,y) \}$ means that $a$ should take a value among those $x$ that are parents of some $y$. Expression~(\ref{f:set}) is defined as the conjunction of:
\begin{eqnarray}
\forall x \ (\varphi(x) \rightarrow f(\V{t})=x \vee f(\V{t})\neq x ) \ \label{f:set1} \\
\neg \exists x \ (\varphi(x) \wedge f(\V{t})=x) \rightarrow \bot \label{f:set2}
\end{eqnarray}

Other typical set constructions can be defined in terms of~(\ref{f:set}):
\begin{eqnarray*}
f(\V{t}) \in \{t'(\V{y}) \ | \ \exists \V{y} \ \varphi(\V{y}) \} & \eqdef & f(\V{t}) \in \{x \ | \ \exists \V{y} \ (\varphi(\V{y}) \wedge t'(\V{y})=x) \} \\
f(\V{t}) \in \{t'_1,\dots,t'_n\} & \eqdef & f(\V{t}) \in \{ x \ | \ t'_1=x \vee \dots \vee t'_n=x\}
\end{eqnarray*}

It must be noticed that variable $x$ in~(\ref{f:set}) is not free, but implicitly quantified and local to this expression. Note that $\varphi(x)$ may contain other quantified and/or free variables. For instance, observe the difference between:
\begin{eqnarray}
Person(y) \rightarrow a(y) \in \{ x \ | \ Parent(x,y) \} \label{f:ex1}\\
Person(y) \rightarrow a(y) \in \{ x \ | \ \exists y \ Parent(x,y) \} \label{f:ex2}
\end{eqnarray}
\noindent In (\ref{f:ex1}) we assign, per each person $y$, one of her parents to $a(y)$, whereas in (\ref{f:ex2}) we are assigning \emph{any} parent as, in fact, we could change the set condition to $\exists z \ Parent(x,z)$.

At a first sight, it could seem that the formula $\exists x (\varphi(x) \wedge f(\V{t})=x)$ could capture the expected meaning of $f(\V{t}) \in  \{ x \ | \ \varphi(x) \}$ in a more direct way. Unfortunately, such a formula would not ``pick'' a value $x$ among those that satisfy $\varphi(x)$. For instance, if we translate $
a \in \{ x \ | \ \exists y \ Parent(x,y) \}$ as $\exists x (\exists y \ Parent(x,y) \wedge a=x)$ would allow the free addition of facts for $Parent(x,y)$. Notice also that a formula like $a \in \{t\}$ is stronger than an assignment $a \is t$ since when $t$ is undefined, the former is always false, regardless the value of $a$ (it would informally correspond to an expression like $a \in \emptyset$).

\section{Logic Programs with Evaluable Functions}
\label{sec:FLP}

In this section we consider a subset of $\QELF$ which corresponds to a certain kind of logic programs that allow evaluable functions but not constructors other than a finite set of Herbrand constants $\C=\C_0$. The interest of this syntactic class is that it can be translated into ground normal logic programs, and so, equilibrium models can be computed by any of the currently available answer set provers. From now on, we assume that any function $f/n$ with arity $n>0$ is evaluable, $f/n \in \F$, and any constant $c$ is a constructor, $c \in \C$, unless we include a declaration $c/0 \in \F$. As usual in logic programming notation, we use in this section capital letters to represent variables.

In what follows we will use the tag `FLP' to refer to functional logic programming definitions, and `LP' to talk about the more restrictive syntax of normal logic programs (without functions). An FLP-\emph{atom} has the form $p(\V{t})$, $t_1=t_2$ or $t_1 \ap t_2$, where $p$ is a predicate name, $\V{t}$ a tuple of terms and $t_1, t_2$ a pair of terms. An FLP-\emph{literal} is an FLP-atom $A$ (called \emph{positive} literal) or its default negation $\neg A$ (called \emph{negative} literal). We call LP-\emph{terms} (resp. LP-\emph{atoms}, resp. LP-\emph{literals}) to those not containing evaluable function symbols.

\begin{definition}[FLP-rule]\label{def:flprule}
An FLP-\emph{rule} is an implication $\alpha \IF \beta$ where $\beta$ (called \emph{body}) is a conjunction of literals, and $\alpha$ (called \emph{head}) has the form of one the following expressions:
\begin{itemize}
\item[(i)] an FLP-atom $p(\V{t})$;
\item[(ii)] the truth constant $\bot$;
\item[(iii)] an \emph{assignment} $f(\V{t}) := t'$ with $f \in \F$;
\item[(iv)] or a \emph{choice} like $f(\V{t}) \in \{ x \ | \ \varphi(x) \}$ with $f \in \F$ and $\varphi(x)$ a conjunction of literals. We call $x$ the \emph{choice variable} and $\varphi(x)$ the choice condition.\qed
\end{itemize}
\end{definition}
 
Function $f$ in (iii) and (iv) is called the \emph{head function}. A \emph{choice rule} is a rule with a \emph{choice} head. A \emph{functional logic program} is a set of FLP-rules. The following is an example of a program in FLP syntax:
  
\begin{example}[Hamiltonian cycles]\label{ex:ham}
A \emph{Hamiltonian cycle} is a cyclic path in a graph that visites all its nodes exactly once. We encode this problem using a function $next(X)$ that specifies which is the next node in the path for node $X$, and $visited(X)$ that keeps track of visited nodes. The program $\Pi_{\ref{ex:ham}}$ consists of the following rules:
\begin{eqnarray}
next(X) \in \{Z \ | \ arc(X,Z)\} \IF node(X)  \label{f:ham2}\\
visited(next(0)) \label{f:ham3}\\
visited(next(X)) \IF visited(X) \label{f:ham4}\\
\bot \IF node(X)\wedge \neg visited(X) \label{f:ham5}
\end{eqnarray}
\noindent where we assume we always have some node $0$, we can call the ``initial'' one.\qed
\end{example}

An LP-\emph{rule} is such that its body exclusively contains LP-literals and its head is either $\bot$ or an LP-atom $p(\V{t})$. An LP-\emph{program} is a set of LP-rules. As LP-rules do not contain evaluable functions, any LP-rule $\alpha \IF \beta$ is simply equivalent to $\beta \rightarrow \alpha$. Thus, an LP-program has the form of a (standard) normal logic program with constraints and without evaluable functions. The absence of evaluable functions guarantees that $\QELF$ and $\QEL$ coincide for this kind of program:

\begin{proposition}
$\QELF$ equilibrium models of an LP-program $\Pi$ correspond to $\QEL$ equilibrium models of $\Pi$.\qed
\end{proposition}

\subsection{Translation to programs without functions}

The translation of an FLP-program $\Pi$ will be done in two steps. In a first step, we will define a $\QEL$ theory $\Gamma(\Pi)$ for a different signature and prove that it is $\SQHTF$ equivalent modulo the original signature. This theory $\Gamma(\Pi)$ is not an LP-program yet, as it will allow existential quantifiers and double negations in the rule bodies. However, these features can be removed in a second step by introducing auxiliary predicates, so that the resulting LP-program preserves strong equivalence wrt equilibrium models (modulo the original signature). We will focus here on the first translation step -- for a description on the transformations removing existential quantifiers and double negations and its complete proof of correctness see~\cite{Cab09}.

The main idea of the translation is that, for each evaluable function $f/n \in \F$ occurring in $\Pi$ we will introduce a predicate like $holds\_f(X_1,\dots,X_n,V)$ in $\tr{\Pi}$, or $holds\_f(\V{X},V)$ for short. The technique of converting a function into a predicate and shifting the function value as an extra argument is well known in Functional Logic Programming and has received the name of \emph{flattening}~\cite{Nai91,Rou94}. Flattening in ASP was also applied for translating the languages in~\cite{CL04,Cab05} into (function-free) logic programs, and in~\cite{LW08} to show that total functions can be removed in favour of predicates. Obviously, once we deal with a predicate, we will need that no two different values are assigned to the same function. This can be simply captured by:
\begin{eqnarray}
\bot & \leftarrow & holds\_f(\V{X},V) \wedge holds\_f(\V{X},W) \wedge \neg (V=W) \label{f:unique}
\end{eqnarray}
\noindent with variables $V, W$ not included in $\V{X}$. 

Given the original signature $\Sigma=\tuple{\C,\F,\P}$ for program $\Pi$, the theory $\Gamma(\Pi)$ will deal with a new signature $\tr{\Sigma}=\tuple{\C,\emptyset,\tr{\P}}$ where $\tr{\P}$ consists of $\P$ plus a new predicate $holds\_f/(n+1)$ per each evaluable function $f/n \in \F$.

\begin{definition}[Correspondence of interpretations]
Given an $HT$ interpretation $I=\tuple{S^\h,S^\t}$ for signature $\Sigma=\tuple{\C,\F,\P}$ we define a \emph{corresponding interpretation} $\tr{I}=\tuple{(\sigma^\h,J^\h),(\sigma^\t,J^\t)}$ for signature $\tr{\Sigma}=\tuple{\C,\emptyset,\tr{\P}}$ so that, for any $f/n \in \F$, any tuple $\V{c}$ of $n$ elements from $\C$, any predicate $p/n \in \P$ and any $w \in \{\h,\t\}$:
\begin{enumerate}
\item $holds\_f(\V{c},d) \in J^w$ iff $\sigma^w(\V{c})=d$ with $d \in \C$.
\item $p(\V{c}) \in J^w$ iff $p(\V{c}) \in I^w$. \qed
\end{enumerate}
\end{definition}

Once (\ref{f:unique}) is fixed, the correspondence between $I$ and $\tr{I}$ is bidirectional:
\begin{proposition}
\label{prop:1}
Given signature $\Sigma=\tuple{\C,\F,\P}$ and an interpretation $J$ for $\Sigma^*$ satisfying {\rm (\ref{f:unique})}, then there exists an interpretation $I$ for $\Sigma$ such that $\tr{I}=J$.
\end{proposition}

\begin{definition}[Translation of terms]
We define the translation of a term $t$ as the triple $\tuple{\tr{t},\Phi(t)}$ where $\tr{t}$ is an LP-term and $\Phi(t)$ is a formula s.t.:
\begin{enumerate}
\item For an LP-term $t$, then $t^* \eqdef t$ and $\Phi(t) \eqdef \top$.

\item When $t$ is like $f(\V{t})$ with $f$ an evaluable function, then
$t^* \eqdef X$ and $\Phi(t) \eqdef \Phi(\V{t}) \wedge holds\_f(\tr{\V{t}},X)$ where $X$ is a new fresh variable and $\Phi(\V{t})$ stands for the conjunction of all $\Phi(t_i)$ for all terms $t_i$ in the tuple $\V{t}$.
\qed
\end{enumerate}
\end{definition}

For 0-ary evaluable functions, we would have that $\V{t}$ is empty -- in this case we just assume that $\Phi(\V{t})=\top$. We introduce now some additional notation. Given a term $t$, $subterms(t)$ denotes all its subterms, including $t$ itself. Given a set of terms $T$, by $\tr{T}$ we mean $\{\tr{t} \ | \ t \in T \}$. If $\rho$ is a replacement of variables by Herbrand constants $\replace{\V{X}}{\V{c}}$, we write $I,w,\rho \models \varphi$ to stand for $I,w \models \varphi\replace{\V{X}}{\V{c}}$. Given a conjunction of literals $B=L_1\wedge \dots \wedge L_n$, we denote $\tr{B}\eqdef \tr{L}_1\wedge \dots \wedge \tr{L}_n$.

\begin{definition}[Translation of literals]\label{def:lit}
The \emph{translation of an atom} (or positive literal) $A$ is a formula $\tr{A}$ defined as follows:
\begin{enumerate}
\item If $A$ has the form $p(\V{t})$, then $\tr{A} \eqdef \exists \V{X} \big(\  p(\tr{\V{t}}) \wedge \Phi(\V{t}) \ \big)$ where $\V{X}$ is the set of new fresh variables in $\tr{subterms(\V{t})}$ (those not occurring in the original literal).
\item If $A$ is like $(t_1=t_2)$, then $\tr{A} \eqdef \exists \V{X}\big(\ \tr{t}_1=\tr{t}_2 \wedge \Phi(t_1) \wedge \Phi(t_2) \ \big)$  where $\V{X}$ is the set of new fresh variables in $\tr{subterms(t_1)} \cup \tr{subterms(t_1)}$.
\item If $A$ is like $(t_1 \ap t_2)$, then $\tr{A} \eqdef \exists \V{X}\big(\ \tr{t}_1 \neq \tr{t}_2 \wedge \Phi(t_1) \wedge \Phi(t_2) \ \big)$  where $\V{X}$ is the set of new fresh variables in $\tr{subterms(t_1)} \cup \tr{subterms(t_1)}$.
\end{enumerate}
\noindent The \emph{translation of a negative literal} $L=\neg A$ is the formula $\tr{L}\eqdef \neg \tr{A}$.\qed
\end{definition}

Notice the difference in the translation of the two kinds of inequalities $\neq$ and $\ap$. For instance, while $f(X)=0$ becomes the negated formula $\neg \exists X' \ (X'=0 \wedge holds\_f(X,X'))$, i.e., either $f$ has no value or it has a non-zero value, the literal $f(X) \ap 0$ becomes the formula $\exists X' \ (X' \neq 0 \wedge holds\_f(X,X'))$ which can be seen as ``positive,'' as negation does not affect to any predicate appart from equality.

\begin{definition}[Translation of rules]
The \emph{translation of an (FLP) rule} $r$ like $H \IF B$ is a conjunction of formulas $\Gamma(r)$ defined as follows:
\begin{enumerate}
\item If $H=\bot$, then $\Gamma(r)$ is the formula $\bot \leftarrow \tr{B}$.
\item If $H$ is like $p(\V{t})$ then $\Gamma(r)$ is the formula $p(\tr{\V{t}}) \leftarrow \Phi(\V{t}) \wedge \tr{B}$
\item If $H$ has the form $f(\V{t}) \is t'$ then $\Gamma(r)$ is the formula\\
$holds\_f(\tr{\V{t}},\tr{t'}) \leftarrow \Phi(\V{t}) \wedge \Phi(t') \wedge \tr{B}
$
\item If $H$ has the form $f(\V{t}) \in \{X \ | \ \varphi(X) \} $ then $\Gamma(r)$ is the conjunction of:
\begin{eqnarray}
holds\_f(\tr{\V{t}},X) \vee \neg holds\_f(\tr{\V{t}},X) \leftarrow  \Phi(\V{t}) \wedge \tr{B} \wedge \tr{\varphi(X)} \label{f:bb1}\\
\bot \leftarrow \neg \exists X (holds\_f(\tr{\V{t}},X) \wedge \tr{\varphi(X)}) \wedge \Phi(\V{t}) \wedge \tr{B} \label{f:bb2}
\end{eqnarray}
\noindent where we assume that, if $X$ happened to occur in $B$, we have previously replaced it in the choice by a new fresh variable symbol, say $\{Y \ | \ \varphi(Y)\}$.
\end{enumerate}
\end{definition}

\begin{definition}[Translation of a program $\Gamma(\Pi)$]
The \emph{translation of an FLP program} $\Pi$ is a theory $\Gamma(\Pi)$ consisting of the union of all $\Gamma(r)$ per each rule $r \in \Pi$ plus, for each evaluable function $f/n$, the schemata (\ref{f:unique}).\qed
\end{definition}

\begin{theorem}[Correctness of $\Gamma(\Pi)$]
\label{th:correct}
For any FLP-program $\Pi$ with signature $\Sigma=\tuple{\C,\F,\P}$ any pair of interpretations $I$ for $\Sigma$ and $J$ for $\tr{\Sigma}$ such that $J=\tr{I}$:
$I,w \models \Pi$ iff $\tr{I},w \models \Gamma(\Pi)$.\qed
\end{theorem}

As an example, the translation of $\Pi_{\ref{ex:ham}}$ is the theory $\Gamma(\Pi_{\ref{ex:ham}})$:
\begin{eqnarray}
holds\_next(X,Z) \vee \neg holds\_next(X,Z) \leftarrow arc(X,Z) \wedge node(X) \label{f:c2}\\
\bot \leftarrow \neg \exists Z (holds\_next(X,Z) \wedge arc(X,Z)) \wedge node(X) \label{f:c3}\\
visited(X) \leftarrow holds\_next(0,X) \label{f:c4}\\
visited(X_2) \leftarrow holds\_next(X,X_2) \wedge visited(X) \label{f:c5}\\
\bot \leftarrow node(X) \wedge \neg visited(X) \label{f:c6}\\
\bot \leftarrow holds\_next(X,V) \wedge holds\_next(X,W) \wedge \neg (V=W) \label{f:c7}
\end{eqnarray}

Of course, $\Gamma(\Pi)$ is not a normal logic program, since in the general case it contains rule heads like $\varphi \vee \neg \varphi$ for some atom $\varphi$, as in \eqref{f:c2}, or expressions in the body like $\exists X \ \varphi(X)$ with $\varphi(X)$ a conjunction of literals, as in \eqref{f:c3}. However, as we said before and is detailed in \cite{Cab09} and~\cite{LP09}, we can always build an LP-program $\tr{\Pi}$ by removing these constructions and introducing new auxiliary predicates. For instance, a formula like $\varphi \vee \neg \varphi \leftarrow \beta$, which is strongly equivalent to a double negation in the body $\varphi  \leftarrow \neg \neg \varphi \wedge \beta$, can be replaced by the pair of rules $(\varphi \leftarrow \neg aux \wedge \beta)$ and $(aux \leftarrow \neg \varphi \wedge \beta)$ where $aux$ is a new auxiliary predicate. Similarly, we can replace a formula $\exists X \ \varphi(X)$ in a rule body by a new auxiliary predicate $aux'$, and include a rule $(aux' \leftarrow \varphi(X))$ for its definition. Of course, the auxiliary predicates must incorporate as arguments all the free variables of the original expression they replace. In our example, the final program $\Pi^*_{\ref{ex:ham}}$ would result from replacing in $\Gamma(\Pi_{\ref{ex:ham}})$ the formula (\ref{f:c2}) by rules:
\begin{eqnarray*}
holds\_next(X,Z) \leftarrow \neg aux(X,Z) \wedge arc(X,Z) \wedge node(X) \\
aux(X,Z) \leftarrow \neg holds\_next(X,Z) \wedge arc(X,Z) \wedge node(X)
\end{eqnarray*}
\noindent and (\ref{f:c3}) by rules:
\begin{eqnarray*}
aux'(X) \leftarrow holds\_next(X,Z) \wedge arc(X,Z) \wedge node(X) \\
\bot \leftarrow \neg aux'(X) \wedge node(X)
\end{eqnarray*}

\subsection{Safety}

As we will translate a set of FLP-rules into a set of LP-rules, when trying to ground the latter we will need to guarantee their domain independence, i.e., that the set of stable models are not affected by extending the set of constants. To this aim we introduce a notion of \emph{safety} for FLP-rules that guarantees the safety of the resulting LP program.

\begin{definition}[Restricted variable]\label{def:rv}
A variable $X$ is said to be \emph{restricted} in a conjunction of literals $\beta$ by a positive literal $A$ in $\beta$ when $X$ occurs in $A$ and one of the following holds:
\begin{enumerate}
\item \label{def:rv_it1} $A$ has the form $p(\V{t})$;
\item $A$ contains a term $f(\V{t})$ and $X$ is one of the arguments in $\V{t}$;
\item $A$ has the form $f(\V{t})=X$ or $X=f(\V{t})$.
\item \label{def:rv_it4} $A$ has the form $X=Y$ or $Y=X$ and, in its turn, $Y$ is restricted by a different positive literal $A'$ in $\beta$.
\end{enumerate}
\noindent We just say that $X$ is \emph{restricted} in $\beta$ if it is restricted by some $A$ in $\beta$.\qed
\end{definition}

As an example, given the conjunction of literals:
\begin{eqnarray}
p(X,f(Y)) \wedge f(Z)\ap W \wedge \neg q(V) \wedge V=Y \label{f:rest}
\end{eqnarray}
\noindent $X$ and $Y$ are restricted by the first literal in \eqref{f:rest}, $Z$ is restricted by the second literal and $V$ is restricted by $V=Y$, since $Y$ is already restricted by the first literal. Similarly, in Example~\ref{ex:ham}, observe that $X$ is restricted in the bodies of \eqref{f:ham2}, \eqref{f:ham4} and \eqref{f:ham5}.

\begin{definition}[Safe rule]\label{def:safe}
A rule $r$, of one of the forms in Definition~\ref{def:flprule}, is said to be \emph{safe} when each variable $X$ occurring in $r$ satisfies:
\begin{enumerate}
\item If $X$ is not a choice variable and is not restricted in the body of $r$ then: 
\begin{itemize}
\item $X$ does not occur in the scope of negation and
\item $X$ is not $t'$ nor one of the arguments in $\V{t}$, in any of the possible forms of the head of $r$.
\end{itemize}
\item If $X$ is a choice variable, then it is restricted in the choice condition $\varphi(x)$.\qed
\end{enumerate}
\end{definition}
For instance, the rules $p(f(X),Y) \IF q(Y)$ and $g \in \{Y \ | \ p(Y) \}$ are safe, whereas the rules $f(Z) \is 0$ or $g \in \{Y \ | \ \neg p(Y) \}$ are not safe. A safe program is a set of safe rules. It can be easily checked that the FLP-program $\Pi_{\ref{ex:ham}}$ in Example~\ref{ex:ham} is safe.

When we restrict our definition of safety to the case of LP-programs, Definition~\ref{def:safe} trivially amounts to the standard concept of safe rule, with a minor exception due to our slightly weakened concept of restricted variable. In particular, when we do not have functions, the standard concept of restricted variable $X$ is requiring that $X$ is the argument of some positive literal formed with a non-equality predicate (Item~1 of Definition~\ref{def:rv}). In our case, we also allow that a variable becomes restricted by an equality atom $X=Y$ (Item~2 of Definition~\ref{def:rv}) in the trivial case where $Y$ is restricted by another positive literal, like in the example:
\begin{eqnarray}
p(X) \leftarrow q(Y) \wedge X=Y \label{f:eqsafe}
\end{eqnarray}
\noindent This rule, for instance, is considered unsafe by current implementations of grounders {\tt DLV}~\cite{LPF06+} and {\tt GrinGo}~\cite{GST07}, although it is strongly equivalent to $p(X) \leftarrow q(X)$ which is obviously safe. Of course, rejecting a rule like~\eqref{f:eqsafe} is not a great restriction, since a programmer would rarely write this kind of redundant code. In our case, however, accepting~\eqref{f:eqsafe} as safe gets a relative importance since bodies like this may easily arise from our automated translation from FLP-programs to LP-programs. Anyway, in order to make rules like \eqref{f:eqsafe} acceptable by current ASP grounders, we assume that our final LP-program $\Pi^*$ can be post-processed to remove these redundant variables by exhausting the rewriting rule:
\begin{eqnarray*}
\frac{\beta \wedge X=Y \rightarrow \alpha}
{\beta\replace{Y}{X} \rightarrow \alpha\replace{Y}{X}}
\end{eqnarray*}

As $\Pi^*$ is free of functions, its set of $\QELF$ equilibrium models coincides with its $\QEL$ equilibrium models, and these in their turn, in the case of safe programs, coincide with the set of stable models of the grounded version of the program $\Pi^*$. So, given the correspondence between $\Pi^*$ and $\Pi$ established in Theorem~\ref{th:correct}, we just remain to prove the following.

\begin{theorem}\label{th:safe}
If $\Pi$ is safe then $\tr{\Pi}$ is safe.\qed
\end{theorem}

\section{Lin and Wang's evaluable total functions}
\label{sec:lin}

As we commented in the introduction, \cite{LW08} introduced a closely related approach for dealing with (evaluable) functions in ASP. We will refer to this approach as \emph{FASP}, taking the name of its associated implementation. In what follows, we show that FASP can be embedded into a subclass of $\QELF$ where all functions are \emph{total}, that is, they satisfy axiom \eqref{f:totalf}. 

FASP formalism is a many-sorted first order language, so that all constants, variables, predicate arguments and function arguments and values belong to a predefined type or sort, containing a finite and non-empty set of elements. 

\begin{definition}[FASP-signature]
A FASP-\emph{signature} has the form $\tuple{\C,\F,\P,\T,\rho}$ where $\C$, $\F$ and $\P$ have the same meaning as before and:
\begin{itemize}
\item $\C=\C_0$ that is, all constructors are $0$-ary (like in our FLP-programs);
\item $\F_0=\emptyset$, that is, there are no $0$-ary evaluable functions;
\item $\T$ is a non-empty finite set of \emph{type names}, at least including $bool$;
\item $\rho$ is a \emph{rank function} that assigns a pair $(\V{T},\tau)$ to each predicate, function and variable in the language, so that $\V{T}$ is a (possibly empty) tuple of type names called the \emph{domain}\footnote{Sometimes also called \emph{arity}, but we prefer here to maintain this name for the \emph{number} of arguments in the tuple.} and $\tau$ is a type name called the \emph{range}, so that, for predicates, $\tau=bool$, and for variables $\V{T}=\epsilon$ (the empty tuple).\qed
\end{itemize}
\end{definition}
We will use the following abbreviations for rank declarations:
\begin{eqnarray}
p & \subseteq & \tau_1 \times \dots \times \tau_n \label{p-rank}\\
f & : & \tau_1 \times \dots \times \tau_n \longrightarrow \tau_{n+1} \label{f-rank} \\
X & : & \tau \label{x-rank}
\end{eqnarray}
\noindent that respectively stand for $\rho(p)=((\tau_1, \dots, \tau_n),bool)$, $\rho(f)=((\tau_1, \dots, \tau_n),\tau_{n+1})$ and $\rho(X)=(\epsilon,\tau)$.

A FASP-\emph{rule} is an expression of the form:
\begin{eqnarray*}
A \leftarrow B_1 \wedge \dots \wedge B_m \wedge \neg C_1 \wedge \dots \wedge \neg C_m
\end{eqnarray*}
\noindent where $A$ is $\bot$ (empty) or a predicate atom and $B_i$, $1 \leq i\leq m$, and $C_j$, $1\leq j \leq n$ are atoms. As a many-sorted formalism, all terms and atoms occurring in the rule are supposed to additionally satisfy a \emph{type coherence} restriction: for short, all arguments of predicates and functions must be of a compatible sort with respect to their rank. In the case of equality, $t_1=t_2$ both $t_1$ and $t_2$ must belong to the same sort.

 A FASP-\emph{program} $\Pi$ is a set of FASP-rules together with a set of \emph{type definitions}, one for each type $\tau$ used in the rules of $\Pi$ and having the form $\tau : \{c_1, \dots, c_n\}$ where the $c_i$ is an enumeration of constant names with $n>0$. The following are a pair of examples extracted from~\cite{LW08}.

\begin{example}[Graph colouring problem]\label{ex:col}
We must assign a colour to each node of a graph so that no two adjacent nodes have the same colour. A possible formalisation in FASP uses a function $clr : node \longrightarrow colour$, a predicate $arc \subseteq node \times node$, a pair of variables $X,Y : node$ and the program $\Pi_{\ref{ex:col}}$ containing the single rule:
\begin{eqnarray}
\bot & \leftarrow & arc(X,Y) \wedge clr(X)=clr(Y) \label{f:colour}
\end{eqnarray}
\qed
\end{example}

\begin{example}[Hamiltonian Cycles in FASP]\label{ex:hamb}
For instance, the Hamiltonian Cycles are encoded in FASP with the program $\Pi_{\ref{ex:hamb}}$ consisting of rules:
\begin{eqnarray}
\bot & \leftarrow & \neg arc(X,next(X)) \label{f:fasp1} \\
visited(next(0)) & & \label{f:fasp2}\\
visited(next(X)) & \leftarrow & visited(X) \label{f:fasp3}\\
\bot & \leftarrow & \neg visited(X) \label{f:fasp4}
\end{eqnarray}
\noindent together with the following domain and range declarations
\[
\begin{array}{rcl@{\hspace{50pt}}rcl}
arc & \subseteq & node \times node &
next & : & node \longrightarrow node \\
visited & \subseteq & node &
X & : & node 
\end{array}\\
\hfill~\Box
\]
\end{example}

\begin{definition}[FASP-Interpretation]
Given a signature $\tuple{\C,\F,\P,\T,\rho}$, a \emph{FASP-interpretation} $S$ is a state $(\sigma,A)$ for $\tuple{\C,\F,\P}$ additionally satisfying:
\begin{itemize}
\item $\sigma(f({\V{c}}))\neq \u$ (functions are total)
\item For each predicate $p$ with domain $p \subseteq \tau_1 \times \dots \times \tau_n$ then, each atom $p(\V{c}) \in A$ satisfies $\V{c} \in \tau_1 \times \dots \times \tau_n$.
\item For each funcion $f$ with domain and range $f : \tau_1 \times \dots \times \tau_n \longrightarrow \tau_{n+1}$, $n>0$, then for any $\V{c} \in \tau_1 \times \dots \times \tau_n$, $\sigma(f(\V{c})) \in \tau_{n+1}$.\qed
\end{itemize}
\end{definition}

Given a FASP program $\Pi$, its grounding contains all type definitions in $\Pi$ plus the rules that are obtained by replacing all variables in the rules of $\Pi$ by the elements in their respective ranges in all the possible ways. Notice that the grounding of $\Pi$ may introduce constant symbols that did not occur in the non-ground rules, but were elements of some type $\tau_i \subseteq \C$ in the signature.

\begin{definition}[Reduction $\Pi^S$]
We define the \emph{reduction} of a (ground) FASP-program $\Pi$ under a FASP-interpretation $S=(\sigma,A)$, written $\Pi^S$, as the set of rules obtained from $\Pi$ by iterating the following transformations:
\begin{itemize}
\item replace each functional term $f(\V{c})$ in a rule by $\sigma(f(\V{c}))$;
\item replace by $\bot$ any equality literal like $c \neq c$ or $c = d$ with constants $c$ and $d$ syntactically different;
\item replace by $\bot$ any body literal $\neg p(\V{c})$ such that $p(\V{c}) \in A$;
\item replace by $\top$ the rest of literals $\neg p(\V{c})$ and the rest of equality literals from the bodies of the remaining rules.\qed
\end{itemize}
\end{definition}
We further assume that rules containing $\bot$ in their body are removed whereas all $\top$ constants are removed from rule bodies.

It is easy to see that the ground program $\Pi^S$ does not contain negation, equality or functions, although it may contain constraints. Let $\Pi^S_{nc}$ be the set of non-constraint rules in $\Pi^S$. This program has a propositional least model, a set of ground atoms we denote as $LM(\Pi^S_{nc})$. 

\begin{definition}[Answer Set]
We say that a FASP-interpretation $S=(\sigma,A)$ is an \emph{answer set} of a (ground) FASP-program $\Pi$ if $A=LM(\Pi^S_{nc})$ and $A$ satisfies all the constraints in $\Pi^S$.\qed
\end{definition}

\subsection{Correspondence to FLP-programs}

It may be noticed that the main syntactic difference between FASP and FLP-programs relies in that the former are many-sorted. To overcome this difficulty, we will introduce sorts in FLP-programs as abbreviations of additional conditions and constraints. To this aim, given a FASP-program, we define the corresponding FLP-program $\hat{\Pi}$ as follows.

For each type declaration $\tau : \{c_1, \dots, c_n\}$ in a FASP-program $\Pi$ we include a new fresh predicate with the same name $\tau$ in the signature of $\hat{\Pi}$ plus the set of FLP-atoms $\tau(c_i)$ for $1\leq i \leq n$. For each function rank declaration like \eqref{f-rank} we include in $\hat{\Pi}$ the rule:
\begin{eqnarray}
f(X_1,\dots,X_n) \in \{ X' \ | \ \tau_{n+1}(X') \} \leftarrow \tau_1(X_1) \wedge \dots \wedge \tau_n(X_n) \label{f:choice}
\end{eqnarray}
\noindent and for any FASP-rule $\alpha \leftarrow \beta$ containing variables $X_1 \dots X_n$ with their respective ranges $\tau_1,\dots,\tau_n$, we include in $\hat{\Pi}$ the FLP-rule:
\begin{eqnarray}
\alpha \leftarrow \beta \wedge \tau_1(X_1) \wedge \dots \wedge \tau_n(X_n) \label{f:varrange}
\end{eqnarray}
For instance, program $\hat{\Pi}_{\ref{ex:col}}$ would correspond to:
\begin{eqnarray}
clr(X) \in \{Y \ | \ colour(Y)\} & \leftarrow & node(X) \label{f:col_c3}\\
\bot & \leftarrow & arc(X,Y) \wedge clr(X)=clr(Y) \nonumber \\
& & \wedge \ node(X) \wedge node(Y) \label{f:col_c4}
\end{eqnarray}
\noindent plus a set of facts for unary predicates $node$ and $colour$. Similarly, $\hat{\Pi}_{\ref{ex:hamb}}$ would consist of:
\begin{eqnarray}
next(X) \in \{Y \ | \ node(Y)\} & \leftarrow & node(X) \\
\bot & \leftarrow & \neg arc(X,next(X)) \wedge node(X) \\
visited(next(0)) & & \\
visited(next(X)) & \leftarrow & visited(X) \wedge node(X)\\
\bot & \leftarrow & \neg visited(X) \wedge node(X) 
\end{eqnarray}

It can be noticed that, for translating FASP-programs, we do not actually need using operator $\IF$ because functions are total (when applied to arguments in their domain). As a second observation, it is easy to see that, since all variables in any FASP program $\Pi$ are sorted, the resulting program $\hat{\Pi}$ will be safe, since any rule where a variable $X$ occurs will include in its body a predicate atom\footnote{In fact, this works in the same way as directive {\tt \#domain} directive in {\tt lparse}, Section 5.5 in~\cite{Syr07}, for declaring sorted variables.} $\tau(X)$. As a result, we can just focus the comparison on the ground versions of $\Pi$ and $\hat{\Pi}$, we respectively denote $grnd(\Pi)$ and $grnd(\hat{\Pi})$. Note that the grounding of a rule like \eqref{f:choice} corresponds to its definition as derived operator in terms of \eqref{f:set1} and \eqref{f:set2}. Furthermore, as types have a finite extension $\tau: \{c_1, \dots, c_n\}$, a formula like $\exists X (\tau(X) \wedge \alpha(X))$ can be unfolded as a finite disjunction $(\tau(c_1) \wedge \alpha(c_1)) \vee \dots \vee (\tau(c_n) \wedge \alpha(c_n))$. 

For instance, the grounding of \eqref{f:col_c3} for $node : \{1,2\}$ and $colour : \{r,g\}$ would contain (among other with false body) the set of rules:
\begin{eqnarray*}
clr(1)=r \vee clr(1)\neq r & \leftarrow & node(1) \wedge colour(r) \\
clr(1)=g \vee clr(1)\neq g & \leftarrow & node(1) \wedge colour(g) \\
\bot & \leftarrow & node(1) \wedge \neg (colour(g) \wedge f(1)=g \vee colour(r) \wedge f(1)=r)\\
clr(2)=r \vee clr(2)\neq r & \leftarrow & node(2) \wedge colour(r) \\
clr(2)=g \vee clr(2)\neq g & \leftarrow & node(2) \wedge colour(g) \\
\bot & \leftarrow & node(2) \wedge \neg (colour(g) \wedge f(2)=g \vee colour(r) \wedge f(2)=r)
\end{eqnarray*}
\noindent which, since the extent of $node$ and $colour$ is fixed, can be further simplified into the equivalent program:
\begin{eqnarray*}
clr(1)=r \vee clr(1)\neq r \\
clr(1)=g \vee clr(1)\neq g \\
\bot & \leftarrow & \neg f(1)=g \wedge \neg f(1)=r\\
clr(2)=r \vee clr(2)\neq r \\
clr(2)=g \vee clr(2)\neq g \\
\bot & \leftarrow & \neg f(2)=g \wedge \neg f(2)=r
\end{eqnarray*}

Generalising this process, the following lemma is relatively simple to check.
\begin{lemma}\label{lem:choicegrnd}
The grounding in $\hat{\Pi}$ of a choice rule like \eqref{f:choice} with respect to FASP program $\Pi$ and signature $\tuple{\C,\F,\P,\T,\rho}$ is equivalent to the set of ground formulas:
\begin{eqnarray}
f(\V{d})=c_i \vee f(\V{d}) \neq c_i \label{f:aa1} \\
\bot \leftarrow f(\V{d})\neq c_1 \wedge \dots \wedge f(\V{d}) \neq c_n \label{f:aa2}
\end{eqnarray}
\noindent for any $1 \leq i \leq n$, being $\tau:\{c_1,\dots,c_n\}$ the range of $f$, and for any $\V{d}$ tuple of constants in $\C$ such that $\V{d}$ belongs to the domain of $f$.\qed
\end{lemma}

After examining the satisfaction of formulas in $\SQHTF$, from this we easily conclude the next result.

\begin{lemma}\label{lem:totfun}
Any $\SQHTF$ interpretation $I=(S^h,S^t)$, with $S^h=(\sigma^h,A^h)$ and $S^t=(\sigma^t,A^t)$, is a model of \eqref{f:aa1} and \eqref{f:aa2} iff $\sigma^h(f(\V{d}))=\sigma^t(f(\V{d}))=c$ being $\V{d}$ a tuple of constants in the domain of $f$, and $c$ some constant in the range of $f$.\qed
\end{lemma}

\begin{lemma}\label{lem:grnd}
Let $\Pi$ be a FASP-program for signature $\tuple{\C,\F,\P,\T,\rho}$ and $I$ any $\SQHTF$ interpretation $(S^h,S)$ with $S^h=(\sigma^h,A^h)$ and $S=(\sigma,A)$. Then $I \models grnd(\hat{\Pi})$ iff $I \models grnd(\Pi)^S$.\qed
\end{lemma}

\begin{theorem}\label{th:lw}
Given a ground FASP-program $\Pi$ for signature $\tuple{\C,\F,\P,\T,\rho}$, $S=(\sigma,A)$ is an answer set for $grnd(\Pi)$ iff $(S,S)$ is an equilibrium model for $grnd(\hat{\Pi})$.
\end{theorem}

\subsection{Some remarks on expressiveness}

At the sight of~\cite{LW08}, the reader may wonder about the real need for partial functions for knowledge representation. For instance, any partial function can be easily encoded as a total one by just adding a new special value (typically called $none$) to denote undefinedness\footnote{In fact,~\cite{LW08} does not specify the way in which, for instance, a division by zero should be treated.}. However, the real difference between FASP and  and $\QELF$ is not so related to totality versus partiality, but has more to do instead with a ``classical'' behaviour (what we called decidable functions) versus a true non-monotonic one. To illustrate this concept, consider  the following example.

\begin{example}[Empty chessboard cells]
When describing a chess ending situation, we want to specify the content of each chessboard cell. Typically, most cells will be $empty$, and in a few cases they will contain a chessman. To this aim, we want to use a function $board(X,Y)$ that specifies the content of a given cell position $X:\{a,\dots,h\}$ and $Y:\{1,\dots,8\}$, and a set of facts to describe the occupied cells, like: $board(a,1)=blkKing, board(b,1)=blkPawn, board(d,3)=whtHorse$, etc.\qed
\end{example}

Typically, when encoding this problem in a relational ASP setting, we would include a rule asserting that all cells are empty \emph{by default}. In a functional setting, this means that we need declaring a \emph{default value} for a given function, something that, as we saw in Section~\ref{sec:der}, can be compactly represented with the rule:
\begin{eqnarray*}
board(X,Y) \is empty & \IF & row(X) \wedge column(Y) \wedge \neg (board(X,Y) \ap empty)
\end{eqnarray*}
\noindent whose informal reading is ``assign an empty content when there is no evidence that the cell is non-empty.'' An important remark is that, although there may exist $\SQHTF$ models in which function $board$ is partial, \emph{this function will be total in any equilibrium model} (for any correct cell position $X,Y$), since the default above cannot leave $board(X,Y)$ undefined. 

On the other hand, a default like this does not seem easily representable in FASP, unless we make use of additional auxiliary predicates, i.e., we end up resorting to the relational fragment of FASP. The reason for this difficulty is that functions are decidable, and so, their value can be defined ``from the start.'' In this way, in FASP, we would have a free choice for selecting \emph{any} value for any function, and then only choices satisfying the rules and constraints eventually lead to an answer set. In our example, this means that if we just enumerate the occupied cells, we would have an answer set for \emph{any} possible combination of contents of the rest of cells, but no way to assume they are empty by default.

A similar difficulty would arise when representing inertia for functions when dealing with an actions and change scenario, something that in $\QELF$ would have a quite natural representation. For instance, if $board$ became a fluent, with a third parameter $I$ for representing a situation number, its inertia could be written as:
\begin{eqnarray*}
board(X,Y,I+1) \is board(X,Y,I) & \IF & \neg ( board(X,Y,I+1) \ap board(X,Y,I) )
\end{eqnarray*}

Finally, $\QELF$ allows a functional interpretation of predicates, as done for instance in~\cite{CL04,Cab05} so that we can define them as functions with a boolean range $\{true,false\}$. As shown in~\cite{CL04}, if we further assert that $false$ is a default value for all boolean functions, we obtain the same expressiveness as standard ASP. To put an example, the program
\begin{eqnarray*}
p & \leftarrow & \neg q \\
q & \leftarrow & r \wedge \neg p \\
r & \leftarrow & \neg s
\end{eqnarray*}
\noindent would be re-encoded using this technique as:
\begin{eqnarray*}
p=true & \leftarrow & q=false \\
q=true & \leftarrow & r=true \wedge q=false\\
r=true & \leftarrow & s=false \\
A \is false & \IF & \neg (A \ap false)
\end{eqnarray*}
\noindent for $A$ varying in $p,q,r,s$, so that the functional equilibrium models of this FLP-program correspond to the (standard) answer sets of the original program. In other words, we can encode full Answer Set Programming by exclusively using (boolean) functions with default values and without resorting to any predicate (excepting equality). In the case of FASP, the impossibility of representing defaults when only dealing with functions (that is, when we supress the use of predicates) would make this enconding to collapse into classical propositional logic.

\section{Related Work}
\label{sec:rel}

The present approach has incorporated many of the ideas previously presented in~\cite{CL04,Cab05}. For instance,~\cite{CL04} can be seen as the fragment of our FLP-programs where we disable the use of predicates and restrict default negation exclusively for specifying default values of functions.

With respect to other logical characterisations of Functional Programming languages, the closest one is perhaps~\cite{GHLR99}, from where we extracted the separation of constructors and evaluable functions. The main difference is that $\QELF$ provides a completely logical description of all operators, allowing an arbitrary syntax (including rules with negation, disjunction in the head, negation and disjunction of rules, etc). Another important difference is that $\QELF$ is constrained to strict functions, while~\cite{GHLR99} is based on non-strict functions. 

Scott's $E$-Logic is not the only choice for logical treatment of partial functions. A related approach is the so-called \emph{Logic of Partial Functions} (LPF)~\cite{BCJ84}. The main difference is that LPF is a three-valued logic -- formulas containing undefined terms have a third, undefined truth value. The relation to (relational) ASP in this way is much more distant than the current approach, since stable models and their logical counterpart, equilibrium models, are two-valued\footnote{Note that in this work we are not considering explicit negation.}.

As for the relation to other approaches exclusively dealing with Herbrand functions~\cite{Syr01,Bon04,SE07} an interesting topic for future study is analysing to which extent they could be captured by $\QEL^=$ semantics, i.e., the fragment of $\QELF$ without evaluable functions. 

\section{Conclusions}
\label{sec:conc}

We can summarize the main contributions of this paper into the introduction of a new language for dealing with functions in ASP and the discussion about several modelling issues not easily solvable within other ASP modelling paradigms. In this way, the paper has tried to clarify some relevant aspects related to the use of functions in ASP for Knowledge Representation. These aspects include the distinction between Herbrand and evaluable (and possibly partial) functions, the concept of definedness, the treatment of equality, the directionality in function assignments or a new nondeterministic choice operation for selecting a function value. 

The functional nature of some predicates is hidden in many ASP domains. When functions are represented in a relational way, we require the continuous addition of constraints for uniqueness of value, and a considerable amount of extra variables to replace the ability of nesting functional terms. All this additional effort may easily become a source for programming errors. 

Although, as we have shown, the proposed approach can be translated into relational ASP and merely considered as \emph{syntactic sugar}, we claim that the use of functions may provide a more natural, compact and readable way of representing many scenarios. The previous experience with a very close language to that of Section~\ref{sec:FLP}, implemented in an online interpreter\footnote{ Available at \url{http://www.dc.fi.udc.es/~cabalar/fal/} } and used for didactic purposes in the past, shows that the functional notation helps the student concentrate on the mathematical definition of the domain to be represented, and forget some of the low level representation tasks, as those commented above, or as the definedness conditions, that must be also considered in the relational representation, but the functional interpreter checks in an automatic way.

We hope that the current approach will help to integrate, in the future, the explicit treatment of arithmetic functions made by some ASP tools, that are currently handled \emph{outside} the formal setting. For instance, the ASP grounder {\tt lparse}\footnote{Available at \url{http://www.tcs.hut.fi/Software/smodels/}.} syntactically accepts a program like $p(div(10,X)) \leftarrow q(X)$ but raises a ``divide by zero'' runtime error if fact $q(0)$ is added to the program. On the other hand, when $div$ is replaced by a non-built-in function symbol, say $f$, the meaning is quite different, and we get $\{p(f(10,0)), q(0)\}$ as a stable model. In this paper we have also identified and separated evaluable and (possibly) partial functions (like $div$ above) from constructors (like $f$ in the previous example).

We have provided a translation of our functional language into normal logic programs to show that: (1) it can be implemented with current ASP solvers; but more important (2) that the proposed semantics is \emph{sensible} with respect to the way in which we usually program in the existing ASP paradigm. This translation has been implemented in a tool called {\tt lppf} (\emph{logic programs with partial functions})\footnote{Available at \url{http://www.equilibriumlogic.net/el/lppf/lppf.pl}}.

For future work, we plan to follow~\cite{LW08} work on loop formulas for converting their programs with total functions into Constraint Satisfaction Problems and extend their work for our functional logic programs. As in~\cite{LW08}, we expect to obtain a reduction on the size of ground functional logic programs, with respect to the size of their relational counterparts.

A topic for future study is the implementation of a solver that directly handles the functional semantics. Other open topics are the axiomatisation of the current logical framework, the addition of a second, explicit (or strong) negation, or the extension of {\tt lppf} to combine evaluable functions with constructors of arity greater than zero, using as a back-end the recently available tool {\tt DLV-complex}\footnote{Available at \url{http://www.mat.unical.it/dlv-complex}}.

\section*{Acknowledgements}
I am especially thankful to Joohyung Lee and Yunsong Meng for pointing out some technical errors in an early version of this work, and to Francisco L\'opez Fraguas for his bibliography guidance on semantics of partial functions in the field of Functional Logic Programming. This research was partially supported by Spanish MEC project TIN-2006-15455-C03-02 and Xunta de Galicia project INCITE08-PXIB105159PR.

\bibliography{refs}

\section*{Appendix. Proofs}

\begin{proof}[Proof of Proposition~\ref{prop:1}]
Then, it suffices with defining $I^w=\{p(\V{c}) \in J^w \ | \ p/n \in \P \}$ and $\sigma^w$ such that, for any evaluable function $f$ and tuple $\V{c}$ in elements of $Terms(\C)$: $\sigma^w(f(\V{c}))=d$ if $holds\_f(\V{c},d) \in J^w$; or $\sigma^w(f(\V{c}))=\u$ otherwise. Note that the latter is well-defined  since (\ref{f:unique}) guarantees that no pair of atoms $holds\_f(\V{c},d)$ and $holds\_f(\V{c},e)$ with $e \neq d$ are included in any $J^w$.  The rest of mapping $\sigma^w$ is built up from its structural definition implied by Condition (ii) in Definition~\ref{def:state}.
\end{proof}

\begin{lemma}
\label{lem:1}
For any term $t$, interpretation $I$ and corresponding interpretation $\tr{I}$, and for any replacement $\rho$ of variables in $\tr{subterms(t)}$ then $\tr{I},w,\rho \models \Phi(t)$ is equivalent to: $I,w,\rho \models E \ t$ and $I,w,\rho \models \tr{(t')}=t'$ for any $t' \in subterms(t)$.
\end{lemma}
\begin{proof}
We proceed by induction. For the base case, when $t$ is an LP-term,
$E \ t$ is valid, and so equivalent to $\top=\Phi(t)$; besides,  $\tr{t}=t$ by definition and $t$ has no subterms. Assume proved for a tuple of terms $\V{t}$ and consider $t=f(\V{t})$. Then note that $\tr{I},w,\rho \models \Phi(t)$ is equivalent to condition (A): $\tr{I},w,\rho \models \Phi(\V{t})$ and $\tr{I},w,\rho \models holds\_f(\tr{\V{t}},X_t)$. Now the first conjunct of (A) is equivalent, by induction, to $I,w,\rho \models E \ \V{t}$ and $I,w \models \tr{(t')}=t'$ for any subterm of $\V{t}$, whereas the second conjunct of (A) is equivalent, by the correspondence between $I$ and $\tr{I}$, to $I,w,\rho \models f(\V{t})=X_t$ provided that we have already obtained $I,w,\rho \models \tr{\V{t}}=\V{t}$. To sum up, (A) is therefore equivalent to $I,w,\rho \models  E \ \V{t} \wedge f(\V{t})=X_t$ and $I,w,\rho \models \tr{(t')}=t'$ for any subterm of $\V{t}$. Since $E \ \V{t} \wedge f(\V{t})=X_t$ is equivalent to $E\ f(\V{t}) \wedge f(\V{t})=X_t$ and this, by definition, is the same than  $E\ t \wedge t=\tr{t}$, we finally obtain $I,w,\rho \models E \ t$ and $I,w,\rho \models \tr{(t')}=t'$ for any subterm of $t$.
\end{proof}

\begin{lemma}
\label{lem:lits}
For any body literal $L$: $\tr{I},w \models \tr{L}$ iff $I,w \models L$.
\end{lemma}
\begin{proof}
Depending on the form of $L$ we have:
\begin{enumerate}
\item If $L$ is some atom $p(\V{t})$, then $\tr{I},w \models \tr{L}$ means that for some substitution $\rho$ of variables in $\tr{subterms(\V{t})}$:  $\tr{I},w,\rho \models  p(\tr{\V{t}})$ and $\tr{I},w,\rho \models \Phi(\V{t})$. By Lemma~\ref{lem:1}, the second conjunct is equivalent to $I,w,\rho \models E \ \V{t}$ and $I,w,\rho \models \tr{t}=t$ for any subterm $t$ of $\V{t}$ (and so of $L$), and in particular $I,w,\rho \models \tr{\V{t}}=\V{t}$. But this means that $\tr{I},w,\rho \models  p(\tr{\V{t}})$ is equivalent to $I,w,\rho \models p(\V{t})$ by the correspondence of $I$ and $\tr{I}$. Since $p(\V{t})$ implies $E \ \V{t}$ we can remove the latter and, as a result, the original condition $\tr{I},w,\rho \models \tr{L}$ is equivalent to $I,w,\rho \models p(\V{t})$ and $I,w,\rho \models \tr{t}=t$ for any subterm $t$ of $L$. As $p(\V{t})$ does not contain variables in $\tr{subterms(\V{t})}$, the previous conditions are equivalent to: $I,w \models p(\V{t})$ and there exists some $\rho$ for which $I,w,\rho \models \tr{t}=t$. But as $I,w \models p(\V{t})$ means that $p(\V{t})$ is defined in $I,w$, the existence of a substitution $\rho$ for variables in $\tr{subterms(\V{t})}$ that satisfies $I,w,\rho \models \tr{t}=t$ for any subterm $t$ of $L$ is guaranteed, and so, is a redundant condition that can be removed.

\item If $L$ has the form $t_1=t_2$ then the proof follows similar steps to case 1.

\item If $L$ has the form $\neg A$, then $\tr{I},w \models \neg \tr{A}$ is equivalent to $\tr{I},\t \not\models \tr{A}$. Applying the proof for cases 1 and 2 to atom $A$, this is equivalent to $I,\t\not\models A$ that is further equivalent to $I,w \models \neg A$.
\end{enumerate}
\end{proof}

Obviously, Lemma~\ref{lem:lits} directly implies that $I,w \models B$ is equivalent to $\tr{I},w \models \tr{B}$.

\begin{lemma}
\label{lem:rules}
$\tr{I},w \models \Gamma(r)$ iff $I,w \models r$.
\end{lemma}
\begin{proof}
If $r=(H \IF B)$, depending on the form of $H$ we have:
\begin{enumerate}
\item If $H=\bot$, is easy to see that $(\bot \IF B)$ is equivalent to $(\bot \leftarrow B)$. Then, $\tr{I},w \models \bot \leftarrow \tr{B}$ $\Leftrightarrow$ $\tr{I},\t \not\models \tr{B}$ $\Leftrightarrow$ (by Lemma~\ref{lem:lits}) $I,\t \not\models B$ $\Leftrightarrow$ $I,w \models \bot \leftarrow B$.

\item If $H$ is like $p(\V{t})$, then $p(\V{t}) \IF B$ is equivalent to $p(\V{t}) \leftarrow B \wedge E \ \V{t}$. Then, $\tr{I},w \models p(\tr{\V{t}}) \leftarrow \Phi(\V{t}) \wedge \tr{B}$ $\Leftrightarrow$ for all $w'\geq w$: if $\tr{I},w' \models \Phi(\V{t}) \wedge \tr{B}$ then $\tr{I},w' \models p(\tr{\V{t}})$ . Let us call (A) to this condition. By Lemma~\ref{lem:lits}, $\tr{I},w' \models \tr{B}$ is equivalent to $I,w' \models B$. Now note that rules are universally quantified. Take any replacement $\rho$ of variables in $\tr{subterms(\V{t})}$. By Lemma~\ref{lem:1}, $\tr{I},w',\rho \models \Phi(\V{t})$ is equivalent to $I,w',\rho \models E \ \V{t}$
and $I,w',\rho \models \tr{t'}=t'$ for any $t' \in subterms(\V{t})$. If this holds, $\tr{I},w',\rho \models p(\tr{\V{t}})$, which coincides with $I,w',\rho \models p(\tr{\V{t}})$, is equivalent to $I,w',\rho \models p(\V{t})$. To sum up, (A) is equivalent to: for all $w'\geq w$, if $I,w',\rho \models B \wedge E \ \V{t}$ then $I,w',\rho \models p(\V{t})$ for any replacement $\rho$. But this is the same than $I,w \models p(\V{t}) \leftarrow B \wedge E \ \V{t}$.

\item If $H$ has the form $f(\V{t}) \is t'$, we may first observe that $(H \IF B)$ is equivalent to $f(\V{t})=t' \leftarrow E \ \V{t} \wedge E \ t' \wedge B$. Then, $\tr{I},w \models holds\_f(\tr{\V{t}},\tr{t'}) \leftarrow \Phi(\V{t}) \wedge \Phi(t') \wedge \tr{B}$ is equivalent to, for any world $w'\geq w$ and any replacement of variables $\rho$: if $\tr{I},w',\rho \models \Phi(\V{t}) \wedge \Phi(t') \wedge \tr{B}$ then $\tr{I},w',\rho \models holds\_f(\tr{\V{t}},\tr{t'})$. By Lemmas~\ref{lem:1} and~\ref{lem:lits}, the antecedent is equivalent to $I,w',\rho \models E \ \V{t} \wedge E \ t' \wedge B$ plus $I,w',\rho \models \tr{k}=k$ for each $k \in subterms(\V{t} \cdot t')$. On the other hand, $\tr{I},w',\rho \models holds\_f(\tr{\V{t}},\tr{t'})$ is equivalent, by correspondence of $I$ and $\tr{I}$, to $I,w',\rho \models f(\tr{\V{t}})=\tr{t'}$ and this, in presence of the equivalent condition for the antecedent we obtained before, is equivalent to $I,w',\rho \models f(\V{t})=t'$. The rest of the proof follows as in the previous case.

\item If $H$ has the form $f(\V{t}) \in \{X \ | \ \varphi(X) \} $ then, after some simple transformations, it can be checked that $(H \IF B)$ is equivalent to the conjunction of the formulas:
\begin{eqnarray}
f(\V{t})=X \vee \neg f(\V{t})=X \leftarrow \varphi(X) \wedge E \ \V{t} \wedge B \label{f:aaa1} \\
\bot \leftarrow \neg \exists X (\varphi(X) \wedge f(\V{t})=X) \wedge      E \ \V{t} \wedge B \label{f:aaa2}
\end{eqnarray}
The proof for this case is tedious, but follows similar steps to the previous two cases. By analogy, it is not difficult to see that $I,w  \models (\ref{f:aaa1})$ iff $\tr{I},w \models (\ref{f:bb1})$ and that $I,w  \models (\ref{f:aaa2})$ iff $\tr{I},w \models (\ref{f:bb2})$.
\end{enumerate}
\end{proof}

\begin{proof}[Proof of Theorem~\ref{th:correct}]
The proof directly follows from Lemma~\ref{lem:rules}.
\end{proof}

For the proof of Theorem~\ref{th:safe} we will show that safety is preserved for the first step of the translation, that is, when the resulting program contains double negation and existential quantifiers in the rule bodies. To this aim, we recall below the definition of safety for rules of this form extracted from~\cite{Cab09}. 

\begin{definition}[Safe rule]\label{def:safe2}
A rule $r: H \leftarrow B$ is said to be safe when both:
\begin{itemize}
\item[a)] Any free variable occurring in $r$ also occurs free and restricted in $\beta$.
\item[b)] For any condition $\exists x \ \varphi$ in $B$, $x$ occurs free and restricted in $\varphi$.\qed
\end{itemize}
\end{definition}

\noindent where the definition of restricted variable in a conjunction of literals is Definition~\ref{def:rv}, but only the applicable items~1 and~2, that do not deal with functions. Note that, for free variables, the above condition means that unrestricted variables cannot occur in the head or negated in the body.

\begin{lemma}\label{lem:rv}
If $X$ is restricted in an FLP-rule conjunction of literals $B$, then $X$ is restricted $B^*$.
\end{lemma}
\begin{proof}
Following Definition~\ref{def:rv} we have four cases:
\begin{enumerate}
\item If $X$ was restricted by some $p(\V{t})$ then $B^*$ will contain a corresponding positive atom $p(\V{t}^*)$ where functional terms have been replaced by auxiliary variables but $X$ still belongs to the tuple $\V{t}^*$.
\item If $X$ was in a term $f(\V{t})$ inside a positive atom in $B$, then $X$ will be included in the corresponding atom $holds\_f(\V{t}^*,Y)$ that will also be positive in $B^*$.
\item It $X$ was in a positive atom $f(\V{t})=X$ (analogously for $X=f(\V{t})$) then the translation will contain an auxiliary variable $Y$ and the pair of positive atoms $Y=X$ and $holds\_f(\V{t}^*,Y)$.
\item If $X$ was in a positive atom $X=Y$ (resp. $Y=X$) and $Y$ was restricted by another different atom, note that $X=Y$ will be preserved in $B^*$ and that we can apply the previous items for concluding that $Y$ is restricted in $B^*$.
\end{enumerate}
\end{proof}

\begin{lemma}\label{lem:1step}
If $\Pi$ is safe then $\Gamma(\Pi)$ is safe.\qed
\end{lemma}
\begin{proof}
We will have two types of variables in $\Gamma(\Pi)$: the original ones in $\Pi$ plus the auxiliary ones introduced in the translation of functional terms. We will show their safety in $\Gamma(\Pi)$ for each case, further distinguishing between choice and non-choice variables, when they belonged to $\Pi$.
\begin{itemize}
\item If $X$ is a variable in some rule $r: H \IF B$ in $\Pi$ and is not a choice variable, we may have that it was restricted in $B$ or not. If it was restricted in $B$, from Lemma~\ref{lem:rv} and the fact that $B^*$ belongs to the bodies of all rules in $\Gamma(\Pi)$, we conclude that $X$ it is also restricted in those rule bodies, and so, $X$ is safe in $\Gamma(r)$. If $X$ was not restricted in the body, as it was safe, it was not in the scope of negation in $\Pi$ and was not $t'$ or one of $\V{t}$ in any of the possible heads in Definition~\ref{def:flprule}. Following the translation, it is easy to see that a variable can only end being in the scope of negation if it already occurred in a negative literal in the body of $r$ or it was one of the arguments in $\V{t}$ in a head of the form $f(\V{t})\in \{Y \ | \ \varphi(Y)\}$, but none of these cases hold. On the other hand, it can also be checked that a variable can end in a head of $\Gamma(r)$ only when it was an element in $\V{t}$ in head like $p(\V{t})$, a head like $f(\V{t})\in \{Y \ | \ \varphi(Y)\}$, or a head like $f(\V{t})=t'$, or $X$ was $t'$ in the last case. But again, none of these cases hold. As a result, $X$ does not occur (free) in the heads of rules in $\Gamma(\Pi)$ nor negated in their bodies.
\item If $X$ is a choice variable in $\Pi$ for some rule with head $f(\V{t})\in \{X \ | \ \varphi(X)\}$, since it was safe, we know that it is restricted in $\varphi(X)$. From Lemma~\ref{lem:rv} we conclude that $X$ is restricted in $\varphi(X)^*$. Now, $\Gamma(\Pi)$ contains the rules \eqref{f:bb1} and \eqref{f:bb2}. In the case of \eqref{f:bb1}, as $\varphi(X)^*$ belongs to the body without being inside an existential quantifier, we immedieatelu conclude that $X$ is restricted in the body, and so is safe in that rule. For \eqref{f:bb2}, we have that $X$ ends being existentially quantified, inside a formula $\exists X ( holds\_f(\V{t}^*,X) \wedge \varphi(X)^* )$, but as $X$ is free and restricted inside the quantified formula, we conclude again that it is safe in the rule.
\item If $X$ is an auxiliary variable, it can only be one of the auxiliary variables $\V{X}$ in Definition~\ref{def:lit} for translation of literals. Note that, when we translate a positive body literal $A$ into $A^*$, the latter will be included in the final rule bodies, whereas it has the form of $\exists \V{X} ( \alpha(\V{X}) )$ and, this is crucial, that $\alpha(\V{X})$ results from translating terms in $A$ and is always a conjunction of positive literals. Thus $\V{X}$ will be restricted in $\alpha(\V{X})$ and thus, these variables will be safe in the result. The same happens for negative literals $\neg A$ and their translation $\neg \exists \V{X} ( \alpha(\V{X}) )$, since safety for existentially quantified variables only requires that they are restricted inside the quantified formula.
\end{itemize}
\end{proof}

\begin{proof}[Proof of Theorem~\ref{th:safe}]
It follows from Lemma~\ref{lem:1step} for $\Gamma(\Pi)$, resulting from the first step of the translation, and from Theorem~7 in~\cite{Cab09} for the second step that eventually yields $\tr{\Pi}$.
\end{proof}

\begin{proof}[Proof of Lemma~\ref{lem:grnd}]
First, we observe that the grounding of FASP-rules yields the same result in $grnd(\Pi)$ and $grnd(\hat{\Pi})$. This is because, for any rule $\alpha \leftarrow \beta$ in $\Pi$, we get a rule \eqref{f:varrange} in $\hat{\Pi}$. But then, after grounding, we can remove those rules in $\hat{\Pi}$ for which $X$ has been replaced by some $c$ not in the range of $X$, since in those cases, there is no head $\tau(c)$ in $grnd(\hat{\Pi})$. Similarly, when $c$ belongs to the range of $X$, $\tau(c)$ will be a fact in $\hat{\Pi}$, and so, it can be removed from the rule body, so that we obtain the same result as directly grounding $\alpha \leftarrow \beta$ in $\Pi$. 

Now, from Lemma~\ref{lem:totfun} we get that $\sigma^h$ and $\sigma$ coincide for the evaluation of functions. Thus, we can replace any functional term $f(\V{c})$ in $grnd(\hat{\Pi})$ by its value $\sigma(f(\V{c}))$. On the other hand, from Proposition~\ref{prop:neg}, we can replace any $\neg \varphi$ such that $I,t \not\models \varphi$ by $\bot$, and any one such that $I,t \not\models \varphi$ by $\top$. Finally, as all function terms in $grnd(\hat{\Pi})$ refer to arguments in the corresponding function domain, equality is always applied to defined terms, and so, it has the same interpretation in $S^h$ and $S$. As a result, $I \models t_1=t_2$ iff $I \models \neg\neg (t_1=t_2)$ and we can replace equality by the corresponding truth constant, as we did for negative literals.
\end{proof}

\begin{proof}[Proof of Theorem~\ref{th:lw}]
For the left to right direction, assume $S$ is answer set for $grnd(\Pi)$ but $(S,S)$ is not equilibrium model of $grnd(\hat{\Pi})$. This means there exists some smaller model $I=(S^h,S)$ of $grnd(\hat{\Pi})$, $S^h=(\sigma^h,A^h)$ that, from Lemma~\ref{lem:totfun}, satisfies $\sigma^h=\sigma$ and for which $A^h \subset A$. From Lemma~\ref{lem:grnd}, $I \models grnd(\hat{\Pi})$ is equivalent to $I \models grnd(\Pi)^S$. Now, as $grnd(\Pi)^S$ does not contain function symbols or negation, it is easy to see that the latter is equivalent to $A^h \models grnd(\Pi)^S$ in propositional logic. But the latter contradicts the fact that $S=(\sigma,A)$ is answer set of $grnd(\Pi)$.

For the right to left direction, assume $(S,S)$ is equilibrium model of $grnd(\hat{\Pi})$ but not an answer set of $grnd(\Pi)$. The latter means there exists some $A' \subset A$ for which $A' \models grnd(\Pi)^S$. But then, we can build the $\SQHTF$ interpretation $I=(S^h,S)$ with $S^h=(\sigma,A')$. As $grnd(\Pi)^S$ does not contain negation or function symbols, $A' \models grnd(\Pi)^S$ implies $I \models grnd(\Pi)^S$ and, in its turn, by Lemma~\ref{lem:grnd}, this is equivalent to $I \models grnd(\hat{\Pi})$. But since $I$ is strictly smaller than $(S,S)$, we get a contradiction with the equilibrium condition for the latter.
\end{proof}

\end{document}